\documentclass[11pt]{article}
\usepackage{graphicx}
\usepackage{amsmath,amssymb}
\usepackage{color,tikz}
\usepackage{subfig}
\usetikzlibrary{patterns}
\usetikzlibrary{decorations.pathmorphing}

\newtheorem{theorem}{Theorem}
{}
\newtheorem{conjecture}{Conjecture}
\newtheorem{lemma}{Lemma}
\newtheorem{observation}{Observation}
\newtheorem{corollary}{Corollary}

\newenvironment{proof}{{\bf Proof}: } {}

\newcommand{\hqed}{\ensuremath{\hfill\Box}} 

\newcommand{\CC}{{\mathcal C}}

\def\qed{\ensuremath{\hfill\Box}}

\newif\ifcomment\commentfalse
\def\commentON{\commenttrue}

\long\outer\def\bc#1\ec{{\ifcomment \sloppy  \textcolor{red}{
{ {#1}}}\fi }}

\long\outer\def\bcr#1\ecr{{\ifcomment \sloppy  \textcolor{purple}{\#  \dotfill
{ {#1}}  \dotfill \#}\fi }}

\long\outer\def\BC#1\EC{{\ifcomment \sloppy \par \textcolor{blue}{\#  \dotfill
{\textsc{#1}} \dotfill \#} \par \fi }}

\long\outer\def\BCR#1\ECR{{\ifcomment \sloppy \par \textcolor{brown}{\#  \dotfill
{\textsc{#1}} \dotfill \#} \par \fi }}

\commentON

\begin{document}

\title{The traveling salesman problem on cubic and subcubic graphs\thanks{This research was partially supported by Tinbergen Institute, the Netherlands and the Natural Sciences and Engineering Research Council of Canada.}
}

\author{Sylvia Boyd\thanks{School of Information Technology and Engineering (SITE), University of Ottawa, Ottawa, Canada. \texttt{sylvia@site.uottawa.ca}} \and Ren{\'e} Sitters\thanks{Department of Operations Research, VU University Amsterdam, The Netherlands \texttt{\{r.a.sitters,suzanne.vander.ster,l.stougie\}@vu.nl}} \and Suzanne van der Ster\footnotemark[3] \and Leen Stougie\footnotemark[3]~\thanks{CWI, Amsterdam, The Netherlands \texttt{stougie@cwi.nl} }}

\maketitle

\begin{abstract}
We study the Travelling Salesman Problem (TSP) on the metric completion of cubic and subcubic graphs, which is known to be NP-hard.  The problem
is of interest because of its relation to the famous 4/3 conjecture for metric TSP, which says that the integrality gap, i.e., the worst case
ratio between the optimal values of the TSP and its linear programming relaxation (the subtour elimination relaxation), is $4/3$. We present the first algorithm for cubic graphs with approximation ratio 4/3. The proof uses polyhedral techniques in a surprising way, which is of independent interest. In fact we prove constructively that for any cubic graph on $n$ vertices a tour of length $4n/3-2$ exists, which also implies the 4/3 conjecture, as an upper bound, for this class of graph-TSP.

 Recently, M\"omke and Svensson presented a randomized algorithm that gives a 1.461-approximation for graph-TSP on general graphs and as a side result a 4/3-approximation algorithm for this problem on subcubic graphs, also settling the 4/3 conjecture for this class of graph-TSP. We will present a way to derandomize their algorithm which leads to a smaller running time than the obvious derandomization. All of the latter also works for multi-graphs.
\end{abstract}

\section{Introduction}\label{sec:intro}

Given a complete undirected graph $G = (V, E)$ with vertex set $V$, $|V| =n$, and edge set $E$, with non-negative edge costs $c \in {\mathbf R}^E$, $c \neq 0$,
the well-known \emph{Traveling Salesman Problem} (TSP) is to find a Hamiltonian cycle in $G$ of minimum cost.  When the costs
satisfy the triangle inequality, i.e. when $c_{ij} + c_{jk} \geq c_{ik}$ for all $i,j,k \in V$, we call the problem
\emph{metric}.  A special case of the metric TSP is the so-called \emph{graph-TSP}, where, given an undirected, unweighted
underlying graph $G=(V,E)$, a complete weighted graph on $V$ is formed by defining the cost between two vertices as the number of
edges on the shortest path between them. This new graph is known as the \emph {metric completion} of $G$. Equivalently, this can be formulated as the problem of finding a spanning Eulerian multi-subgraph $H=(V, E')$ of $G$ with a minimum number of edges, which can be transformed into a graph-TSP tour of $G$ of cost $|E'|$ and vice versa.

The TSP is well-known to be NP-hard \cite{HamNPC}, even for the special cases of graph-TSP. As noticed in \cite{GKP}, APX-hardness follows rather
straightforwardly from the APX-hardness of (weighted) graphs with edges of length 1 or 2 ((1,2)-TSP) (Papadimitriou and Yannakakis \cite{PY}),
even if the maximum degree is $6$.

In general, the TSP cannot be approximated in polynomial time to within any constant unless $P = NP$, however for the metric TSP there exists the elegant
algorithm due to Christofides \cite{chr} from $1976$ which gives a $3/2$-approximation.  Surprisingly, in over three decades no one has found an
approximation algorithm which improves upon this bound of $3/2$, and the quest for finding such
improvements is one of the most challenging research questions in combinatorial optimization.

A related approach for finding approximated TSP solutions is to study the \emph{integrality gap} $\alpha(TSP)$, which is the
worst-case ratio between the optimal solution for the TSP problem and the optimal solution to its linear programming relaxation,
the so-called \emph{Subtour Elimination Relaxation} (henceforth SER) (see \cite{gen} for more details). The value $\alpha(TSP)$
gives one measure of the quality of the lower bound provided by SER for the TSP. Moreover, a polynomial-time constructive proof
for value $\alpha(TSP)$ would provide an $\alpha(TSP)$-approximation algorithm for the TSP.

For metric TSP, it is known that $\alpha(TSP)$ is at most $3/2$ (see Shmoys and Williamson \cite{sw}, Wolsey \cite{w}), and is
at least $4/3$. A ratio of $4/3$ is reached asymptotically by the family of graph-TSP problems consisting of two vertices joined
by three paths of length $k$; see also ~\cite{gen} for a similar family of graphs giving this ratio. However, the exact value of
$\alpha(TSP)$ is not known, and there is the following well-known conjecture, which dates back to the early 1980's:
\begin{conjecture} \label{mainconj}
For the metric TSP, the integrality gap $\alpha(TSP)$ for SER is $4/3$.
\end{conjecture}

\noindent As with the quest to improve upon Christofides' algorithm, the quest to prove or disprove this conjecture has been
open for almost 30 years, with very little progress made.

A graph $G=(V,E)$ is \emph{cubic} if all of its vertices have degree 3, and \emph{subcubic} if they have degree at most 3.
A \emph{multigraph} is one in which multiple copies of edges  (i.e. parallel edges) are allowed between vertices (but loops are not allowed) and a graph is called \emph{simple} if there are no multiple copies of edges. A {\em cycle} in a graph is a closed path having no repetition of vertices.   A {\em cycle cover} (also sometimes referred to as
a 2-factor or a perfect 2-matching) of $G$ is a set of vertex disjoint cycles that together span all vertices of $G$. A {\em perfect matching} $M$ of a graph $G$ is a set of vertex-disjoint edges of $G$ that together span all vertices of
$G$.

In this paper we study the graph-TSP problem on cubic and subcubic graphs.  Note that the graphs in the family described above
giving a worst-case ratio of $4/3$ for $\alpha(TSP)$ are graph-TSPs on bridgeless subcubic graphs. Also, solving the graph-TSP on such graphs would solve the problem of deciding whether a given bridgeless cubic graph $G$ has a
Hamilton cycle, which is known to be NP-complete, even if $G$ is also planar (Garey  et al. \cite{Gar}) or bipartite
(Akiyama  et al. \cite{Aki}). In \cite{CsabaKK2002} there is an unproven claim that (1,2)-TSP is APX-hard when the graph of
edges of length $1$ is cubic, which would imply APX-hardness of graph-TSP on cubic graphs.
Also note that the $3/2$ ratio of Christofides' algorithm is tight for cubic graph-TSP (see \cite{bsss}).

In $2005$, Gamarnik et al. in \cite{GLS} provided the first approximation improvement over Christofides' algorithm for
graph-TSP on 3-edge connected cubic graphs. They provide a polynomial-time
algorithm that finds a Hamilton cycle of cost at most $\tau n$ for $\tau = (3/2 - 5/389) \approx 1.487$.  Since $n$ is a lower bound for the
optimal value for graph-TSP on such graphs, as well as the associated $SER$\footnotemark \footnotetext{To see that $n$ is a lower bound for SER,
sum all of the so-called "degree constraints" for SER.  Dividing the result by $2$ shows that the sum of the edge variables in
any feasible SER solution equals n.}, for any value of $\tau$, this results in a $\tau$-approximation for the
graph-TSP, as well as proves that the integrality gap $\alpha(TSP)$ is at most $\tau$ for such problems.

\smallskip
Only recently the work by Gamarnik et al. has been succeeded by a sudden outburst of results on the approximation of graph-TSP and its SER, which we discuss below.

In 2009 and 2010, polynomial-time algorithms that find triangle- and square-free cycle covers for cubic 3-edge connected graphs have been developed (see \cite{BV},\cite{BoydIT2010} and \cite{HartvigsenLi2009}). These papers do not explicitly study the graph-TSP problem, but as a by-product, these algorithms provide a cycle cover with at most $n/5$ cycles, and thus give a $(1.4n - 2)$-approximation using an approach that we explain below under the name Approach 1.

We made the next improvement (see Boyd et al. \cite{bsss}) by showing that every bridgeless cubic graph has a TSP-tour of length at most $4n/3-2$ when $n\geq 6$. This was the first result which showed that Conjecture \ref{mainconj} is true for graph-TSP, as an upper bound, on cubic bridgeless graphs and it automatically implies a $4/3$-approximation algorithm for this class of graph-TSP. The results extend to all cubic graphs. They have appeared in a preliminary form in \cite{bsss}. The proof of the $4n/3-2$-bound uses polyhedral techniques in a surprising way, which may be more widely applicable. We present a complete proof of the result in Section 2. Just like Garmanik et al. we make use of the following well-known
theorem due to Petersen \cite{Petersen}:

\begin{theorem}  ({Petersen \cite{Petersen}}).\label{thm:peterson}
Any bridgeless cubic graph can be partitioned into a cycle cover and a perfect matching.
\end{theorem}

\noindent The obvious approach that follows from this theorem is:

\smallskip
\noindent {\bf {Approach 1}}:  Given a cubic bridgeless graph $G$ with $n$ vertices, if one can find a cycle cover of $G$ with at most $k$ cycles, then by contracting the cycles, adding a doubled spanning tree in
the resulting graph and uncontracting the cycles, one would obtain a spanning  Eulerian multi-subgraph of $G$ with no more than $n+ 2(k -1) = n + 2k - 2$ edges.

\bigskip
\noindent Approach 1 may exclude the optimal solution. For example, consider the Petersen graph.  Using this approach, the smallest possible solution will have 12 edges, however, there exists a solution with 11 edges, found by taking a 9-cycle plus a single vertex joined by two parallel edges.
Any spanning Eulerian subgraph can be formed by taking a set of cycles and singleton vertices and connecting everything by a doubled tree. Hence, the following approach (which yields an optimal solution in the case of the Petersen graph) can be used:

\smallskip
\noindent {\bf {Approach 2}}: Given a cubic bridgeless graph $G = (V,E)$, if one can find a set $S \subset V$  of singletons and a cycle cover of $G\backslash S$ with at most $k$ cycles, then one can obtain a spanning Eulerian multi-subgraph of $G$ with no more than $(n-|S|) + 2(k + |S| - 1) = n +2k + |S| - 2$ edges.\bigskip

We show that there exists a set of singletons and a cycle cover for Approach 2 such that $|S| + 2k$ is at most $n/3$ for $n\geq 6$, thus obtaining the result. The construction of such a set starts from a cycle cover for $G$, formed by deleting a perfect matching from the graph (cf. Theorem~\ref{thm:peterson}). Local manipulations of the cycle cover leads to larger cycles and singleton vertices. The same approach is used by Garmanik et al.. However, they use only one perfect matching to get to their result while we select a set of polynomially-many perfect matchings for $G$ such that a convex combination of the matchings gives every edge a weight of $1/3$. We prove that the sets of cycles and singletons obtained after local manipulations have the acclaimed size on average in this convex combination. Finding this set of perfect matchings dominates the complexity of the algorithm and uses the $O(n^6)$-time algorithm of Barahona \cite{Bar} to achieve this task.

In \cite{bsss} we also show a bound of  $(7n/5-4/5)$ on the length of a graph-TSP tour for subcubic bridgeless graphs. We conjectured that the true bound should be $(4n/3-2/3)$, which is equal to the lower bound we established for this class of graphs. For reasons that become clear below we do not give the details of this result here but instead refer to the extended version \cite{bsssweb} of \cite{bsss} for its proof.

A little bit later than our work, but independent of it, Aggarwal et al. ~\cite{AggarwalGG2011} announced an
alternative $4n/3$ approximation for 3-edge connected cubic graphs only, but with a simpler algorithm. Their algorithm is based on the idea of finding a triangle- and square-free cycle cover, then shrinking and "splitting off" certain 5-cycles in the cover.

Again, more or less simultaneously, Gharan et al. \cite{GharanSS2010}
announced a randomized $(3/2-\epsilon)$-approximation for graph-TSP for some $\epsilon>0$, which is the first polynomial-time algorithm with an approximation ratio strictly less than 3/2 for graph-TSP on general graphs. Their approach is very different from the one presented here.

Very recently, M\"omke and Svensson \cite{momke} came up with a powerful new approach, which enabled them to prove a $1.461$-approximation for graph-TSP for general graphs. In the context of the present paper it is interesting that their approach led to a bound of $(4n/3 -2/3)$ on the graph-TSP tour for all subcubic bridgeless graphs, thus improving upon our above mentioned $(7n/5 - 4/5)$ bound and settling our conjecture affirmatively.

Viewed in a slightly different way than M\"omke and Svensson present it, for cubic graphs their algorithm is based on the following approach:

\smallskip

\noindent {\bf {Approach 3}}: Given a cubic bridgeless graph $G=(V,E)$, if one can find, in polynomial time,  a spanning tree $T^*$ of $G$, and a perfect matching $M^*$ of $G$, such that $|M^* \cap T^*|$ is at most $p$, and by taking $E$, doubling the edges in $M^*\cap T^*$, and removing the edges in $M^* \setminus T^*$, one obtains a spanning Eulerian multi-subgraph $H$, with edge set $E(H)$, then
\begin{eqnarray*}
\left| {E(H)} \right| &= &\left| {E} \right| + \left| {M^*\cap T^*} \right| - \left| {M^* \setminus T^*} \right| \\
&=& 3n/2 + \left| {M^*\cap T^*} \right| - (n/2 - \left| {M^*\cap T^*} \right|) \\
&=& n + 2\left| {M^*\cap T^*} \right| \\
&\leq& n + 2p.
\end{eqnarray*}
M\"omke and Svensson use this approach in a randomized algorithm, using the convex combination of the perfect matchings mentioned in our approach as a probability distribution. This leads to an obvious derandomization, by just considering all extreme points in the convex combination. In Section \ref{sec:approach3} we take their approach, explain it in a slightly different way and design a way of derandomizing it for graph-TSP on subcubic bridgeless multigraphs, which leads to a faster algorithm $\--$ it has complexity $O(n^2log n)$ rather than $O(n^6)$.  We remark that this can be extended to also provide a faster derandomization for the general graph-TSP case (see Section \ref{sec:3.1}).

As a side result of M\"omke and Svensson's result, hence also of ours, for any given cubic bridgeless multigraph $G$ we obtain a cycle cover of $G$ with at most $\lfloor n/6+2/3 \rfloor$ cycles. In fact, for cubic graphs any solution found by Approach 3 is a solution for Approach 1 and vice versa:
Given a perfect matching and spanning tree with intersection of size $p$, then removing the matching leaves a cycle cover with  at most $p+1$ cycles.
Vice versa, given a cycle cover with  $p+1$ cycles, then we can remove all but $p$ of the matching edges such that the graph stays connected.

Thus far, all results we have mentioned have dealt with bridgeless graphs. In Section \ref{sec:bridges} we show how bridges are easily
incorporated to achieve the same performance guarantees.

We conclude this section with a survey of some of the other relevant literature. Grigni et al. \cite{GKP} give a polynomial-time approximation scheme
(PTAS) for graph-TSP on planar graphs (this was later extended to a PTAS for the weighted planar graph-TSP by Arora et al. \cite{AGKKW}).
For graph $G$ containing a cycle cover with no triangles, Fotakis and Spirakis \cite{FS} show that graph-TSP is approximable in polynomial time
within a factor of $17/12 \approx 1.417$ if $G$ has diameter $4$ (i.e. the longest path has length 4), and within $7/5 = 1.4$ if $G$ has diameter
3.  For graphs that do not contain a triangle-free cycle cover they show that if $G$ has diameter 3, then it is approximable in polynomial time
within a factor of $22/15 \approx 1.467$. For graphs with diameter 2 (i.e. TSP(1,2)), a $7/6 \approx 1.167$-approximation for graph-TSP was
achieved by Papadimitriou and Yannakakis \cite{PY}, and improved to  $8/7 \approx 1.143$ by Berman and Karpinski \cite{BK}.

\section{The first $4n/3$-approximation result for bridgeless cubic graphs}\label{sec:IPCOproof}
In this section, we will prove the following:

\begin{theorem}\label{thm:tsptour2ECcubic}
Every bridgeless simple cubic graph $G=(V,E)$ with $n \geq 6$ has a graph-TSP tour of length at most $\frac{4}{3}n-2$.
\end{theorem}

\noindent We begin by giving some definitions, and preliminary results.

For any vertex subset
$S\subseteq V$, $\delta(S)\subseteq E$, defined as  the set of edges connecting $S$ and $V \setminus S$, is called the {\em cut} induced by $S$. A
cut of cardinality $k$ is called a \emph{$k$-cut} if it is minimal in the sense that it does not contain any cut as a proper subset. A \emph{k-cycle} is a cycle containing $k$ edges, and a \emph{chord} of a cycle of $G$ is
an edge not in the cycle, but with both ends $u$ and $v$ in the cycle.  An {\em
Eulerian subgraph} of $G$ is a connected subgraph where multiple copies of the edges are allowed, and all vertices have even
degree.  Note that such a subgraph has an Eulerian tour of length equal to its number of edges, which can be "short-cut" into a TSP tour of the same length for the associated graph-TSP problem.


As mentioned in Section \ref{sec:intro}, Petersen~\cite{Petersen} states that every bridgeless cubic graph contains a perfect matching. Thus the
edges of any bridgeless cubic graph can be partitioned into a perfect matching and an associated cycle cover. This idea is
important for our main theorem, and we give a useful strengthened form of it below in Lemma \ref{lem:ccofmatchings}.

For any edge set $F \subseteq E$, the \emph{incidence vector of $F$} is the vector $\chi^F \in \lbrace 0,1 \rbrace^E$ defined by
$\chi^F_e =1$ if $e \in F$, and $0$ otherwise. For any edge set $F \subseteq E$ and $x \in \mathbf{R}^E$, let $x(F) = \sum_{e \in F}x_e$.
Given graph $G$, the associated \emph{perfect matching polytope}, $P^{M}(G)$, is the convex hull of all incidence vectors of the
perfect matchings of $G$, which Edmonds \cite{edmonds} shows to be given by:
\begin{eqnarray*}
x(\delta(v)) = 1, &\quad&  \forall v\in V, \\
x(\delta(S)) \geq 1, &\quad& \forall S\subset V,\ |S|\ {\rm odd}, \\
0\leq x_e \leq 1, &\quad& \forall e\in E.
\end{eqnarray*}

Using this linear description and similar methods to those found in \cite{KKN} and \cite{NP}, we have the following strengthened
form of Petersen's Theorem,  in which we use the notion of a {\em 3-cut perfect matching}, which is a perfect matching that intersects every 3-cut of the graph in exactly one edge:
\begin{lemma}\label{lem:ccofmatchings}
Let $G = (V,E)$ be a bridgeless cubic graph and let $x^* = {1\over 3}\chi^{E}$.  Then $x^*$ can be expressed as a convex
combination of incidence vectors of $3$-cut perfect matchings, i.e. there exist $3$-cut perfect matchings $M_i$, $i = 1, 2,...,
k$ of $G$ and positive real numbers $\lambda_i$, $i = 1, 2,..., k$ such that
\begin{eqnarray}\label{convexcombo}
x^*= \sum_{i=1}^k \lambda_i(\chi^{M_i}) \mbox{ and } \sum_{i=1}^k \lambda_i = 1.
\end{eqnarray}
\end{lemma}
\begin{proof}
Since both sides of any 2-cut in a cubic graph have an even number of vertices, it is easily verified that $x^*$  satisfies the linear description
above, and thus lies in $P^M(G)$.  It follows that $x^*$ can be expressed as a convex combination of perfect matchings of $G$, i.e. there exist
perfect matchings $M_i$, $i = 1, 2,..., k$ of $G$ and positive real numbers $\lambda_i$, $i = 1, 2,..., k$ such that~\eqref{convexcombo} holds.

To see that each perfect matching in~\eqref{convexcombo} is a $3$-cut perfect matching, consider any $3$-cut $\delta(S) = \{e_1, e_2, e_3\}$ of
$G$. Since each side of a 3-cut of any cubic graph must contain an odd number of vertices, any perfect matching must contain 1 or 3 edges of
$\delta(S)$. Let ${\cal M}_0$ be the set of perfect matchings from (\ref{convexcombo}) that contain all $3$ edges of the cut, and let ${\cal
M}_j$, $j = 1, 2, 3$ be the sets of perfect matchings that contain edge $e_j$. Define $\alpha_j= \sum_{M_i\in {\cal M}_j}\lambda_i$, $j=0,1,2,3$.
Then
\begin{eqnarray*}
\alpha_0+\alpha_1+\alpha_2+\alpha_3 = 1,\
\alpha_0+\alpha_1=1/3 \mbox{,    } \alpha_0+\alpha_2=1/3 \mbox{,    } \alpha_0+\alpha_3=1/3,
\end{eqnarray*}
which implies $\alpha_0=0$.$\hqed$
\end{proof}
\smallskip

The perfect matchings $M_i, i = 1, 2, ... k$ of Lemma \ref{lem:ccofmatchings} will be used in the proof of our main theorem in
the next section. Note that Barahona \cite{Bar} provides an algorithm to find for any point in $P^M(G)$ a set of perfect
matchings for expressing the point as a convex combination of their incidence vectors in $O(n^6)$ time, and with $k \leq 7n/2
-1$, for any graph $G$.

The idea we will use in the proof of our main theorem is as follows: By Petersen's Theorem we know we can always find a cycle cover of $G$.
Suppose that we can find such a cycle cover that has no more than $n/6$ cycles. Then, contracting the cycles, adding a doubled spanning tree in
the resulting graph and uncontracting the cycles would yield a spanning Eulerian subgraph with no more than $n+ 2(n/6 -1) = 4n/3 - 2$ edges. Together with the
obvious lower bound of $n$ on the length of any optimal graph-TSP tour, this yields an approximation ratio of $4/3$.  However, such a cycle cover does
not always exist (for example, consider the Petersen graph). Therefore, we take the $k$ cycle covers associated with the 3-cut matchings of
Lemma \ref{lem:ccofmatchings} and combine their smaller cycles into larger cycles or Eulerian subgraphs, such as to obtain $k$ covers of $G$ with
Eulerian subgraphs which, together with the double spanning tree, result in $k$ spanning Eulerian subgraphs of $G$ having an average number of edges of at most $4n/3$. Unless stated otherwise, an Eulerian subgraph is connected. As mentioned in the introduction, we may see each of these Eulerian subgraphs as cycles and singleton vertices connected by a doubled tree. For the ease of analysis we shall not make these decompositions explicit.
For the construction of larger Eulerian subgraphs the following lemma will be useful.
\begin{lemma}\label{lem:merge graphs}
Let $H_1$ and $H_2$ be connected Eulerian subgraphs of a (sub)cubic graph such that $H_1$ and $H_2$ have at least two vertices
in common and let $H_3$ be the {\em sum} of $H_1$ and $H_2$, i.e., the union of their vertices and the sum of their edges,
possibly giving rise to parallel edges. Then we can remove two edges from $H_3$ such that it stays connected and Eulerian.
\end{lemma}
\begin{proof}
Let $u$ and $v$ be in both subgraphs. The edge set of $H_3$ can be partitioned into edge-disjoint $(u,v)$-walks $P_1,P_2, P_3, P_4$. Vertex $u$
must have two parallel edges which are on different paths, say $e_1\in P_1$ and $e_2\in P_2$. When we remove $e_1$ and $e_2$ then the graph stays
Eulerian. Moreover, it stays connected since $u$ and $v$ are still connected by $P_3$ and $P_4$ and, clearly, each vertex on $P_1$ and $P_2$
remains connected to either $u$ or $v$.$\hqed$
\end{proof}\smallskip

The following lemma, which applies to any graph, allows us to preprocess our graph by removing certain subgraphs.
\begin{lemma}\label{lem:2cut}
Assume that removing edges $u'u''$ and $v'v''$ from graph $G=(V,E)$ breaks it into two graphs $G'=(V',E')$ and $G''=(V'',E'')$
with $u',v' \in V'$, and $u''v''\in V''$ and such that:
\begin{enumerate}
\item  $u'v'\in E$ and $u''v''\notin E$.
\item  there is spanning Eulerian subgraph $T'$ of $G'$ with at most $4|V'|/3-2$ edges.
\item  there is a spanning Eulerian subgraph $T''$ of $G''\cup u''v''$ with at most $4|V''|/3-2$ edges.
\end{enumerate}
Then there is a spanning Eulerian subgraph $T$ of $G$ with at most $4|V|/3-2$ edges.
\end{lemma}
\begin{proof} If $T''$ does not use edge $u''v''$
then we take edge $u'u''$ doubled and add subgraph $T'$. If $T''$ uses edge $u''v''$ once then we remove it and add edges $u'u''$, $v'v''$ and $u'v'$
and subgraph $T'$.  If $T''$ uses edge $u''v''$ twice then we remove both copies and add edge $u'u''$ doubled, $v'v''$ doubled, and subgraph $T'$.$\hqed$
\end{proof}\smallskip

\begin{figure}\label{fig:removedsubgraph}
\center
\begin{tikzpicture}[scale=1.1]

\draw (0,0) node(u3) [circle,draw] { };
\draw (1,0) node(u2) [circle,draw] { };
\draw (2,0) node(u1) [circle,draw] { };
\draw (3,0) node(u0) [circle,draw] { };
\draw (5,0) node(v0) [circle,draw] { };
\draw (6,0) node(v1) [circle,draw] { };
\draw (7,0) node(v2) [circle,draw] { };
\draw (8,0) node(v3) [circle,draw] { };
\draw (4,1) node(a) [circle,draw] { };
\draw (4,-1) node(b) [circle,draw] { };

\node at (0,-0.3) {$u_3$};
\node at (1,-0.3) {$u_2$};
\node at (2,-0.3) {$u_1$};
\node at (3,-0.3) { $u_0$};
\node at (5,-0.3) { $v_0$};
\node at (6,-0.3) { $v_1$};
\node at (7,-0.3) { $v_2$};
\node at (8,-0.3) {$v_3$};
\node at (4.3,1) { $a$};
\node at (4.3,-1) { $b$};

\draw (u3)--(u2)--(u1)--(u0);
\draw (v3)--(v2)--(v1)--(v0);
\draw (u0)--(a)--(v0)--(b)--(u0);
\draw (a)--(b);

\draw (-0.5,0.5) node(x) { };
\draw (-0.5,-0.5) node(y) { };
\draw (y)--(u3)--(x);
\draw (8.5,0.5) node(x) { };
\draw (8.5,-0.5) node(y) { };
\draw (y)--(v3)--(x);

\draw (node cs:name=u2) .. controls +(1 ,2) and +(-1,2) ..(node cs:name=v2);
\draw (node cs:name=u1) .. controls +(0.8 ,1.7) and +(-0.8,1.7) ..(node cs:name=v1);

\draw [line width=0.5pt, style=dotted] (4,0) ellipse (3.5cm and 1.7cm);
\node at (6.5,-0.8) {$G'$};
\node at (7.5,-0.8) {$G''$};

\end{tikzpicture}
\caption{In this $p$-rainbow example, $p=2$ and  $u'=u_2,u''=u_3,v'=v_2$, and $v''=v_3$.  }
\end{figure}
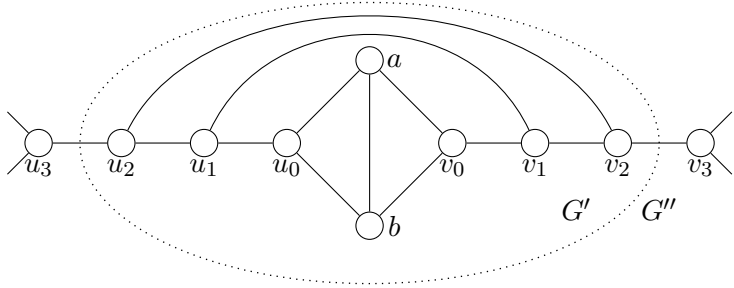

We use Lemma~\ref{lem:2cut} to remove all subgraphs of the form shown in Figure~1, which we call a $\emph{p-rainbow}$ subgraph. In such subgraphs
there is a path $u_0,u_1,\dots,u_{p+1}$ and path $v_0,v_1,\dots,v_{p+1}$ for some $p\ge 1$, and a 4-cycle $u_0,a,v_0,b$ with chord $ab$.
Furthermore, there are edges $u_iv_i$ for each $i\in \{1,2,\dots,p\}$ but there is no edge between $u_{p+1}$ and $v_{p+1}$. The figure shows a
$p$-rainbow for $p=2$. For general $p$, the $2$-cut of Lemma~\ref{lem:2cut} is given by $u'=u_p$, $u''=u_{p+1}$, $v'=v_p$, and $v''=v_{p+1}$. If
$G$ contains a $p$-rainbow $G'$, $p \geq 1$, then we remove $G'$ and add edge $u''v''$ to the remaining graph $G''$. Note that $G''$ is also a
simple bridgeless cubic graph.  We repeat this until there are no more $p$-rainbows in $G''$ for any $p \geq 1$. If the final remaining graph
$G''$ has at least 6 vertices, then assuming $G''$ has a spanning Eulerian subgraph with at most $4/3|V''|-2$ edges, we can apply Lemma~\ref{lem:2cut} repeatedly to
obtain such a subgraph of length at most $4n/3 -2$ for the original graph $G$. If the final remaining graph $G''$ has less than 6 vertices, then it
must have 4 vertices, since it is cubic, hence it forms a complete graph on 4 vertices.  In this case we take the Hamilton path from $u''$ to
$v''$ in $G''$ and match it with the Hamilton path of the $p$-rainbow that goes from $u_p$ to $v_p$ to obtain a Hamilton cycle of the graph $G''$
with the edge $u''v''$ replaced by the $p$-rainbow.  We can then apply Lemma \ref{lem:2cut} repeatedly to  obtain a  spanning Eulerian subgraph of $G$ with at most $4n/3 -2$ edges.


\subsubsection*{Proof of Theorem~\ref{thm:tsptour2ECcubic}.} By the above discussion, we assume that there are no $p$-rainbow subgraphs in $G$.
By Lemma \ref{lem:ccofmatchings} there exist 3-cut perfect matchings $M_1,\ldots,M_k$ and positive real numbers $\lambda_1,\ldots,\lambda_k$ such
that $\sum_{i=1}^k\lambda_i=1$ and $\frac{1}{3}\chi^{E} =\sum_{i=1}^k \lambda_i(\chi^{M_i})$. Let $\CC_1,\ldots,\CC_k$ be the cycle covers of $G$
corresponding to $M_1, M_2, ... M_k$. Since each $M_i$ is a 3-cut perfect matching, each $\CC_i$ intersects each 3-cut of $G$ in exactly 2 edges,
and hence contains neither a $3$-cycle nor a $5$-cycle with a chord.

If some $\CC_i$ has no more than $n/6$ cycles, then we are done, by the argument given earlier. Otherwise we manipulate each of
the cycle covers by operations (i) and (ii) below, which we will show to be well-defined. First operation (i) will be performed
as long as possible. Then operation (ii) will be performed as long as possible.
\begin{itemize}
\item[(i)] If two cycles $C_i$ and $C_j$ of the cycle cover intersect a (chordless) cycle $C$ of length $4$ in $G$ (the original graph) then combine them into a single cycle on $V(C_i)\cup V(C_j)$.

\end{itemize}
The details of operation (i) are as follows:  Assume that $u_1u_2$ and $v_1v_2$ are matching edges on $C$ and
$u_1v_1$ is an edge of $C_i$ and $u_2v_2$ is an edge of $C_j$. Deleting the latter two edges and inserting the former two yields
a single cycle of length equal to the sum of the lengths of $C_i$ and $C_j$. Notice that operation (i) always leads to cycles of
length at least 8. Hence after operation (i) is finished we still have a cycle cover. Operation (ii) below combines cycles into
Eulerian subgraphs and subsequently Eulerian subgraphs into larger Eulerian subgraphs, turning the cycle covers into Eulerian
subgraph covers. Both types of cover we call simply a {\em cover} and their elements (cycles and Eulerian subgraphs) we call
{\em components}.
\begin{itemize}
\item[(ii)] If two components $\gamma_i$ and $\gamma_j$ of the cycle cover or the Eulerian subgraph cover, each having at least 5 vertices,
intersect a (chordless) cycle $C$ of length $5$ in $G$ (the original graph) then combine them into a single Eulerian
subgraph where  the number of edges is 1 plus the number of edges of $\gamma_i$ and $\gamma_j$.

\end{itemize}
The details of operation (ii) are as follows. First note that for any cycle $C$, its vertex set $V(C)$ has the following
(trivial) property:
\\[2mm]
$\mathcal{P}$: Each $v\in V(C)$ has at least two other vertices $u,w \in V(C)$ such that $vu\in E$ and $vw\in E$.\\[2mm]
If two vertex sets both satisfy $\mathcal{P}$ then their union also satisfies $\mathcal{P}$.  Since the vertex set of each
component $\gamma$ constructed by operations (i) or (ii) is a result of taking unions of vertex sets of cycles, each such
$\gamma$ has property $\mathcal{P}$. In particular, since $G$ is cubic, this implies that the two components $\gamma_i$ and
$\gamma_j$ share 2 and 3 vertices with $C$, respectively (note that they cannot each share exactly 2 vertices, as this would
imply that a vertex of $C$ is not included in the cover). We first merge $\gamma_1$ and $C$ as in Lemma~\ref{lem:merge graphs}
and remove $2$ edges, and then merge the result with $\gamma_2$, again removing $2$ edges. Altogether we added the 5 edges of
$C$ and removed $4$ edges.

Operation (ii) leads to Eulerian subgraphs with at least 10 vertices. Thus, any Eulerian subgraph with at most 9 vertices is a
cycle. At the completion of operations (i) and (ii), let the resulting Eulerian subgraph covers be $\Gamma_1,\ldots,\Gamma_k$.

Given $\Gamma_1,\ldots,\Gamma_k$, we bound for each vertex its average contribution to the number of edges in the Eulerian subgraphs weighted by
the $\lambda_i$'s. We define the contribution of a vertex $v$ which in cover $\Gamma_i$ lies on an Eulerian subgraph with $\ell$
edges and $h$ vertices as $z_i(v) = \frac{\ell+2}{h}$; the $2$ in the numerator is added for the cost of the double edge to
connect the component to the others in final spanning Eulerian subgraph. Note that $\sum_{v \in V} z_i(v)$ is equal to the number of edges in $\Gamma_i$, plus 2.  The average contribution of $v$ over all covers is $z(v) = \sum_i \lambda_i
z_i(v)$. When summing this over all vertices $v$ we obtain the average length of the spanning Eulerian subgraphs plus 2. We will show that
$z(v)\leq 4/3$ $\forall v\in V$.

\begin{observation} \label{claim:longcycles} For any vertex $v$ and $i\in\{1,2,\dots,k\}$, the contribution $z_i(v)$ is
\begin{enumerate}
\item[(a)] at most $\frac{h+2}{h}$, where $h = \min\{ t, 10\} $ and $v$ is on a cycle of length $t$ in $\CC_i$ or after
operation (i).
\item[(b)] at most $13/10$ if operation (ii) was applied to some component containing $v$.
\end{enumerate}
\end{observation}
\noindent\textit{Proof (Observation~\ref{claim:longcycles})}. Assume that $v$ is on a Eulerian subgraph $\gamma$ in $\Gamma_i$ of $g$ vertices.
First we prove $(b)$. If operation (ii) was applied to some component containing $v$, then vertex $v$ was on a cycle of length at least $5$ after
operation (i). Each application of (ii) adds at least 5 vertices to the component of $v$. Hence, the number of times that (ii) was applied to the
component of $v$ is at most $g/5-1$. Since each application adds exactly one edge, the number of edges in $\gamma$ is at most $g+g/5-1$. Hence,
\[z_i(v) \leq \frac{g+g/5+1}{g}=\frac{12}{10}+\frac{1}{g}\le \frac{13}{10}.\]
We use a similar argument to prove ($a$). Clearly, $g\ge h$. If $\gamma$ is a cycle then the contribution of $v$ in $\Gamma_i$
is $(g+2)/g\le (h+2)/h$ and ($a$) is true. If $\gamma$ is not a cycle then this  Eulerian subgraph  was composed by operation
(ii) applied to cycles, each of length at least 5 and one of these had length at least $h$. Hence, the number of these cycles is
at most $1+(g-h)/5$. Since every application of operation (ii) adds one edge extra, the  number of edges in $\gamma$ is at most
$g+(g-h)/5$. Hence, since $h\le 10$,
\[ z_i(v) \leq \frac{g+(g-h)/5+2}{g}\le \frac{g+(g-h)/(h/2)+2}{g}=\frac{h+2}{h}.\]\qed

Note the subtleties in Observation~\ref{claim:longcycles}: If $v$ is on a cycle of length $t$ in $\CC_i$ or after operation (i), and  $t \leq 10$,
then ($a$) says that $z_i(v)$ is at most $(t+2)/t$.  If $t > 10$, then ($a$) says that its contribution is at most $12/10$. And finally, if $t$ is
$5$ or $6$ and we know that operation (ii) was applied to some component containing $v$, then ($b$) allows us to improve the upper bound on
$z_i(v)$ to $13/10$ (for other values of $t$, ($b$) does not give an improvement).

From now on we fix any vertex $v$. Suppose that there is no $\ell$ such that $v$ is on a 4-cycle or a 5-cycle of $\Gamma_\ell$. Then using
Observation~\ref{claim:longcycles}, we have $z_i(v) \le \max\{ 8/6,13/10\} = 4/3$ for every cover $\Gamma_i$, and thus $z(v)\leq 4/3$ and we are
done.

Now suppose there exists an $\ell$ such that $v$ is on a 4-cycle $C$ of $\Gamma_\ell$. Then $C$ must be present in $\CC_{\ell}$ as well. First
assume that $C$ is chordless in $G$. Then all four edges adjacent to $C$ are in the set $M_\ell$.
\begin{observation}
\label{obs:3path} For any pair of vertices on a chordless cycle of $G$ that appears in any $\CC_i$, any path between the two
that does not intersect the cycle has length at least $3$.
\end{observation}
We partition the set $\CC_1,\ldots,\CC_k$  according to the way the corresponding $M_i$'s intersect the cycle $C$. Define sets
$X_0, X_1, X_2$ where $X_j = \{ i \mid |C \cap M_i| = j \}$ for $j=0,1,2$. Let $x_t=\sum_{i\in X_t}\lambda_i$, $t=0,1,2$.
Clearly $x_0+x_1 +x_2=1$. Since each of the four edges adjacent to $C$ receives total weight $1/3$ in the matchings, we have
that $4x_0+2x_1=4/3 \Rightarrow x_0=1/3-x_1/2$. Since each of the edges of $C$ receives total weight $1/3$ in the matchings,
$x_1+2x_2=4/3 \Rightarrow x_2=2/3-x_1/2$.

Clearly, for any $i\in X_0$, $v$ lies on cycle $C$  in $\CC_i$, and thus by Observation~\ref{claim:longcycles}($a$), $z_i(v) \leq 6/4$. By
Observation~\ref{obs:3path}, for any $i \in X_1$, $v$ lies on a cycle of length at least 6 in $\CC_i$, and thus by Observation
\ref{claim:longcycles}($a$), $z_i(v) \leq 8/6$. For any $i\in X_2$, if $C$ is intersected by one cycle in $\CC_i$, then this cycle has length at
least 8 by Observation~\ref{obs:3path}. If for $i\in X_2$, $C$ is intersected by two cycles of length at least 4 each, then, after performing
operation (i), $v$ will be on a cycle of length at least 8. Thus using Observation~\ref{claim:longcycles}($a$) one more time, we obtain
\begin{eqnarray*}
z(v)&\le &x_06/4+x_18/6+x_210/8\\
&=& (1/3-x_1/2)6/4+x_18/6+(2/3-x_1/2)10/8\\
&=& 4/3+x_1(8/6-6/8-10/16)=4/3-x_1/24\le 4/3.
 \end{eqnarray*}
We prove now that $z(v)\le 4/3$ also if $C$ is a 4-cycle with a chord. Let us call the vertices on the cycle $u_0,a,v_0,b$, let $ab$ be the chord,
and $v$ is any of the four vertices. If $u_0v_0 \in E$, then $G=K_4$ (the complete graph on 4 vertices), contradicting the assumption that $n \geq
6$.  Thus edges $u_0u_1$ and $v_0v_1$ exist, with $u_1,v_1\notin C$. Notice that $u_1\neq v_1$ since otherwise $G$ would contain a bridge,
contradicting 2-connectedness. Let $C'$ be the cycle containing $v$ in some cycle cover $\CC_i$. If $C'$ does not contain edge $u_0u_1$ then
$C'=C$. If, on the other hand, $u_0u_1\in C'$ then also $v_0v_1\in C'$ and $ab\in C'$.  Note that $u_1v_1 \notin E$ since otherwise we have a
$p$-rainbow subgraph as in Figure~1, and we are assuming that we do not have any such subgraphs. Consequently, $C'$ cannot have length exactly 6.
It also cannot have length 7 since then a 3-cut with 3 matching edges would occur. Therefore, any cycle containing $u_0u_1$ has length at least 8.
Applying Observation~\ref{claim:longcycles}($a$) twice we conclude that $z(v)\leq 1/3 \cdot 6/4 + 2/3 \cdot 10/8 = 4/3$.

Now assume there exists a (chordless) 5-cycle $C$ containing $v$ in some $\Gamma_\ell$. Note that we can assume that no $w \in C$ is on a 4-cycle
of $G$, otherwise operation (i) would have been applied and the component of $v$ in $\Gamma_\ell$ would have size larger than 5. Note further that
$C$ is present in $\CC_\ell$ as well. The proof for this case is rather similar to the case for the chordless 4-cycle. Let $X_j$ be the set $\{
i\mid |C \cap M_i| = j\}$, for $j=0,1,2$. Let $x_t=\sum_{i\in X_t}\lambda_i$, $t=0,1,2$. Again, we have $x_0+x_1+x_2=1$.  Clearly, for any $i\in
X_0$, $v$ lies on $C$ in $\CC_i$ and for $i \in X_1$ $v$ lies on a cycle of length at least $7$ by Observation~\ref{obs:3path}. Hence, by
Observation~\ref{claim:longcycles}($a$) we have $z_i(v)\le 7/5$ for $i\in X_0$ and $z_i(v)\le 9/7$ for $i\in X_1$. For any $i \in X_2$ there are
two possibilities: Either $C$ is intersected by one cycle in $\CC_i$, which, by Observation~\ref{obs:3path}, has length at least 9, or $C$ is
intersected in $\CC_i$ by two cycles, say $C_1$ and $C_2$. In the first case we have $z_i(v)\le 11/9$ by Observation~\ref{claim:longcycles}($a$).
In the second case, as argued before, we can assume that no $w\in C$ is on a 4-cycle of $G$. Hence, $C_1$ and $C_2$ each have at least 5 vertices
and operation (ii) will be applied, unless $C_1$ and $C_2$ end up in one large cycle by operation $(i)$. In the first case we apply
Observation~\ref{claim:longcycles}($b$) and get $z_i(v)\le 13/10$, and in the second case we apply Observation~\ref{claim:longcycles}($a$):
$z_i(v)\le 12/10$. Hence, for any $i \in X_2$ we have $z_i(v)\le \max\{ 11/9,12/10,13/10\}=13/10$.
  \begin{eqnarray*}
  z(v)&\le & x_0 7/5+x_1 9/7 +x_2 13/10\\
&\le & x_0 7/5+x_1 13/10 +x_2 13/10\\
   &=& x_07/5+(1-x_0)13/10 = 13/10 + x_0 1/10 \\
   &\le& 13/10+ 1/30=4/3.
  \end{eqnarray*}\hqed

As previously mentioned, Barahona \cite{Bar} provides a polynomial-time algorithm which finds a set of at most $7n/2 -1$ perfect matchings such
that $\frac{1}{3} \chi^E$ can be expressed as a convex combination of the incidence vectors of these matchings. This algorithm runs in $O(n^6)$
time.  As shown in the proof of Lemma \ref{lem:ccofmatchings}, these matchings will automatically be 3-cut perfect matchings. Once we have this
set of perfect matchings then applying operations (i) and (ii) on the corresponding cycle covers gives at least one tour of length at most
$4n/3-2$ according to the above theorem.  As any tour has length at least $n$ for graph-TSP, we have the following approximation result:
\begin{corollary}\label{4over3}
For graph-TSP on simple bridgeless cubic graphs there exist a polynomial-time $4/3$ approximation algorithm.
\end{corollary}
\noindent As $n$ is a lower bound on the value of SER for graph-TSP it also follows that, as an upper bound, Conjecture
\ref{mainconj} is true for this class of problems, i.e.,
\begin{corollary}
For graph-TSP on simple bridgeless cubic graphs the integrality gap for SER is at most $4/3$.
\end{corollary}
\noindent We remark that the largest ratio we found so far for $\alpha(TSP)$ on simple bridgeless cubic examples is $7/6$  (see Section~\ref{sec:epilogue}).

\section{Subcubic bridgeless graphs: A comment on the M\"omke-Svensson $4n/3$-approximation}\label{sec:approach3}

As discussed in Section \ref{sec:intro}, M\"omke and Svensson \cite{momke} present a randomized algorithm (obviously derandomizable) which is $1.461$-approximate for graph-TSP for general graphs, and gives a bound of $(4n/3 -2/3)$ on the graph-TSP tour for all subcubic bridgeless graphs. Their method is different from all previous methods in that it is based on detecting a set of removable edges $R$ of which some are paired, in the sense that each of them can be removed but not both of them can be removed.

We describe here how detecting the set $R$ and the pairing works. We will then see that this works out particularly nicely for cubic graphs and allows us to derandomize the algorithm, giving a considerable reduction in running time.

The search for $R$ starts by finding a depth first search tree $T$ of the graph $G=(V,E)$ using any vertex $r$ as the root. For the moment, consider the edges of $T$ to be directed away from $r$.  The set of remaining edges is denoted by $B$. They are back edges which are directed towards the root $r$.  By the properties of depth first search trees, each back edge $b=xy$ forms a unique directed cycle together with the path from $y$ to $x$ on $T$.
Let $t_b$ be the unique edge in $T$ on that cycle whose tail is incident with the head $y$ of $b$, and let $T_B = \{t_b: b \in B\}$. Choose the set of removable edges to be
$R=B\cup T_B$.

We make each arc $e\in T_B$ part of a pair in $P$: its partner is chosen arbitrarily from amongst the back edges $b\in B$ such that $e=t_b$. If we think of everything undirected again then notice that, given a pair
$\{ b \in B, t_b \in T_B\}$ in $P$, we have that $(T\setminus t_b) \cup b$ forms a different spanning tree of $G$. In fact, essentially M\"omke and Svensson show indirectly that any number of the $t_b$ and $b$ partnered edge pairs can be swapped, and the result will still be another spanning tree.

\begin{lemma}{\bf \cite{momke}}\label{lem:stillspanningtree}
Let $T_J$ be a subset of the edges in $T_B$, with corresponding partner back edges $J\subseteq B$. Let $T^*$ be the result of taking spanning tree $T$, removing the edges of $T_J$, and adding the edges of $J$.  Then $T^*$ is also a spanning tree of $G$. \qed
\end{lemma}

\begin{lemma}{\bf \cite{momke}}\label{lem:connected}
Let $K$ be a subset of $E$ such that for every pair $\{ b\in B,t_b \in T_B\}$ in $P$ at most one of $t_b$ and its partner back edge $b$ is in $K$. Then $G$ with the edges in $K \cap R$ removed from $E$ is connected. \qed
\end{lemma}

\subsection{Cubic bridgeless multigraphs.}\label{sec:3.1}

Now we turn to the case in which $G$ is a cubic bridgeless multigraph. Notice that each vertex of the depth first search tree $T$ has 0 or 1 back edges directed into it, except for the root $r$, which has exactly 2 back edges directed into it.  This means that every back edge in $B$ is a partner edge for some edge in $T_B$, except for one back edge at the root, and thus $\left|{T_B} \right| = \left| {B} \right| - 1$. So for $G$ we have

 \begin{eqnarray}\label{countpairset}
 \left|{R} \right| = 2\left|{B} \right| -1= 2\left(\left| {E} \right| - (n - 1)\right) - 1 =2(3n/2 - (n - 1)) - 1 = n + 1.
 \end{eqnarray}

We are now ready to establish a bound of $4n/3 - 2/3$ on the length of an optimal TSP-tour in any bridgeless cubic multigraph on $n$ vertices. We do so in a constructive way.

We start by assigning weights to the edges in $G$ based on $R$:
\begin{eqnarray}\label{costs}
c_e = \left \{\begin{array}{rl}
                   -1 & \mbox{ if }e \in R, \\
                    1 & \mbox{ if }e \in E\setminus R.
                   \end{array}
           \right.
\end{eqnarray}
\noindent Then using (\ref{countpairset}), we have
\begin{eqnarray}\label{costE}
 c(E) = -\left| {R} \right| + (\left| {E} \right| -\left| {R} \right|) = 3n/2 - 2(n + 1) = -(n/2 + 2).
\end{eqnarray}

Let $M^*$ be a minimum weight perfect matching for $G$ w.r.t. edge weights $c$.  We use the following theorem, due to Pulleyblank and Naddef \cite{NP} (note that it is not clear in their paper if this theorem is stated for multigraphs or simple graphs, but it trivially follows from the proof they provide that it is also true for multigraphs):
\begin{theorem}  (Naddef and Pulleyblank, Theorem 4 \cite{NP}).\label{thm:NadPul}
  Let $G = (V,E)$ be a $k$-regular $(k-1)$-edge connected multigraph for which $\left| {V} \right|$  is even. Then
  \begin{itemize}
\item[$(i)$] Every edge of G belongs to a perfect matching,
\item[$(ii)$] For any real vector $w = (w_j: j \in E)$ of edge weights, there is perfect matching $M$ of $G$ such that $w(M) \leq w(E)/k$. \qed
\end{itemize}
\end{theorem}

\noindent By Theorem \ref{thm:NadPul} we know that $c(M^*) \leq c(E)/3$, and thus by (\ref{costE}),
\begin{eqnarray} \label{costM}
c(M^*) \leq  -(n/6 + 2/3).
\end{eqnarray}

Now consider the graph $H$ we obtain by taking $G$, removing the edges of $M^* \cap R$, and adding the edges of $M^* \setminus R$. The resulting graph has even degree everywhere, and by Lemma \ref{lem:connected} it is connected.  Thus H is a spanning Eulerian subgraph of $G$. The number of edges in $H$ is
\begin{eqnarray*}
|E|  + |M^* \setminus R| - |M^* \cap R| = |E| + c(M^*) & \leq & 3n/2 - (n/6 + 2/3) \\ & = & 4n/3 - 2/3,
\end{eqnarray*}
\noindent as required.

\begin{figure}
\label{fig:Petersen4}

\subfloat[ ][Petersen graph $G=(V,E)$ ($n=10$).]
{\begin{tikzpicture}[scale=0.7]

\draw (1.55,0) node(d) [circle, draw] { };
\draw (6.55,0) node(c) [circle, draw] { };
\draw (8.1,4.75) node(b) [circle, draw] { };
\draw (0,4.75) node(e) [circle, draw] { };
\draw (4.05,7.68) node(a) [circle, draw] { };

\draw (2.73,1.62) node(i) [circle, draw] { };
\draw (5.37,1.62) node(h) [circle, draw] { };
\draw (6.2,4.13) node(g) [circle, draw] { };
\draw (1.9,4.13) node(j) [circle, draw] { };
\draw (4.05,5.68) node(f) [circle, draw] { };

\draw (4.55, 7.68) node() {$a$};
\draw (8.6, 4.75) node() {$b$};
\draw (7.05, 0) node() {$c$};
\draw (1.05, 0) node() {$d$};
\draw (-0.5, 4.75) node() {$e$};

\draw (4.55,5.68) node() {$f$};
\draw (6.7,3.9) node() {$g$};
\draw (5.85, 1.8) node() {$h$};
\draw (3.2, 1.4) node() {$i$};
\draw (1.4,3.9) node() {$j$};

\draw (a)--(b)--(c)--(d)--(e)--(a);
\draw (f)--(h)--(j)--(g)--(i)--(f);
\draw (a)--(f);
\draw (b)--(g);
\draw (c)--(h);
\draw (d)--(i);
\draw (e)--(j);
\end{tikzpicture}}
\qquad \subfloat[][Depth-first-search tree $T$, back edges $B$, edge-set $T_B$.]
{\begin{tikzpicture}[scale=0.85]

\draw (2,0) node(a) [circle,draw] { };
\draw (2,1.5) node(b) [circle,draw] { };
\draw (2,3) node(c) [circle,draw] { };
\draw (2,4.5) node(h) [circle,draw] { };
\draw (2,6) node(j) [circle,draw] { };
\draw (2,7.5) node(g) [circle,draw] { };
\draw (2,9) node(i) [circle,draw] { };
\draw (0.6,9.5) node(d) [circle,draw] { };
\draw (-0.8,10) node(e) [circle,draw] { };
\draw (3.4,9.5) node(f) [circle,draw] { };

\draw (2.3,-0.3) node() {$a$};
\draw (2.3, 1.2) node() {$b$};
\draw (2.3, 2.7) node() {$c$};
\draw (2.3, 4.2) node() {$h$};
\draw (2.3, 5.7) node() {$j$};
\draw (2.3, 7.8) node() {$g$};
\draw (2.3, 8.7) node() {$i$};
\draw (1, 9.7) node() {$d$};
\draw (-0.4,10.2) node() {$e$};
\draw (3.8, 9.7) node() {$f$};

\draw[->, ultra thick] (a)--(b);
\draw[->, ultra thick] (b)--(c);
\draw[->, ultra thick] (c)--(h);
\draw[->, ultra thick] (h)--(j);
\draw[->, ultra thick] (j)--(g);
\draw[->] (g)--(i);
\draw[->] (i)--(d);
\draw[->] (d)--(e);
\draw[->] (i)--(f);

\draw[->, dashed] (d) .. controls (0,7) and (0.7,3.5) .. (c);
\draw[->, dashed] (e) .. controls (-0.4,8) and (1,6.5) .. (j);
\draw[->, dashed] (e) .. controls (-2,6.5) and (0.5,0.5) .. (a);
\draw[->, dashed] (g) .. controls (3.5,6) and (3.5,3) .. (b);
\draw[->, dashed] (f) .. controls (3.5,7.5) and (3,5) .. (h);
\draw[->, dashed] (f) .. controls (5,7) and (4.5,0.5) .. (a);

\draw (4.5,1.5) node() {$B$};
\draw[dashed] (5,1.5)--(6.5,1.5);
\draw (4.5,0.5) node() {$T_B$};
\draw[ultra thick] (5,0.5)--(6.5,0.5);
\end{tikzpicture}}

\subfloat[][Edge weights $c$ and minimum weight perfect matching $M^*$, $c(M^*) = -3$.]
{\begin{tikzpicture}[scale=0.85]

\draw (2,0) node(a) [circle,draw] { };
\draw (2,1.5) node(b) [circle,draw] { };
\draw (2,3) node(c) [circle,draw] { };
\draw (2,4.5) node(h) [circle,draw] { };
\draw (2,6) node(j) [circle,draw] { };
\draw (2,7.5) node(g) [circle,draw] { };
\draw (2,9) node(i) [circle,draw] { };
\draw (0.6,9.5) node(d) [circle,draw] { };
\draw (-0.8,10) node(e) [circle,draw] { };
\draw (3.4,9.5) node(f) [circle,draw] { };

\draw (2.3,-0.3) node() {$a$};
\draw (2.3, 1.2) node() {$b$};
\draw (2.3, 2.7) node() {$c$};
\draw (2.3, 4.2) node() {$h$};
\draw (2.3, 5.7) node() {$j$};
\draw (2.3, 7.8) node() {$g$};
\draw (2.3, 8.7) node() {$i$};
\draw (1, 9.7) node() {$d$};
\draw (-0.4,10.2) node() {$e$};
\draw (3.8, 9.7) node() {$f$};

\draw (a)--(b);
\draw (b)--(c);
\draw decorate [decoration={name=snake, amplitude=1.5pt, segment length=5pt}] {(c)--(h)};
\draw (h)--(j);
\draw (j)--(g);
\draw (g)--(i);
\draw decorate [decoration={name=snake, amplitude=1.5pt, segment length=5pt}] {(i)--(d)};
\draw (d)--(e);
\draw (i)--(f);

\draw (d) .. controls (0,7) and (0.7,3.5) .. (c);
\draw decorate [decoration={name=snake, amplitude=1.5pt, segment length=5pt}] {(e) .. controls (-0.4,8) and (1,6.5) .. (j)};
\draw (e) .. controls (-2,6.5) and (0.5,0.5) .. (a);
\draw decorate [decoration={name=snake, amplitude=1.5pt, segment length=5pt}] {(g) .. controls (3.5,6) and (3.5,3) .. (b)};
\draw (f) .. controls (3.5,7.5) and (3,5) .. (h);
\draw decorate [decoration={name=snake, amplitude=1.5pt, segment length=5pt}] {(f) .. controls (5,7) and (4.5,0.5) .. (a)};

\draw (1.7, 0.75) node() {$-1$};
\draw (1.7, 2.25) node() {$-1$};
\draw (1.7, 3.75) node() {$-1$};
\draw (1.7, 5.25) node() {$-1$};
\draw (1.7, 6.75) node() {$-1$};

\draw (1.8, 8.25) node() {$1$};
\draw (1.4, 9.6) node() {$1$};
\draw (0, 10.1) node() {$1$};
\draw (2.6, 9.6) node() {$1$};

\draw (-0.4,2.5) node() {$-1$}; 
\draw (0.5,4.4) node() {$-1$}; 
\draw (-0.5,7.9) node() {$-1$}; 
\draw (3,8) node() {$-1$}; 
\draw (4.8,6) node() {$-1$}; 
\draw (3.6,4.5) node() {$-1$}; 

\draw (4.5,1.5) node() {$M^*$};
\draw decorate [decoration={name=snake, amplitude=1.5pt, segment length=5pt}] {(5,1.5)--(6.5,1.5)};
\end{tikzpicture}}
\qquad \subfloat[][Spanning Eulerian multi-subgraph $H$, $|E(H)| = 12.$]
{\begin{tikzpicture}[scale=0.7]

\draw (1.55,0) node(d) [circle, draw] { };
\draw (6.55,0) node(c) [circle, draw] { };
\draw (8.1,4.75) node(b) [circle, draw] { };
\draw (0,4.75) node(e) [circle, draw] { };
\draw (4.05,7.68) node(a) [circle, draw] { };

\draw (2.73,1.62) node(i) [circle, draw] { };
\draw (5.37,1.62) node(h) [circle, draw] { };
\draw (6.2,4.13) node(g) [circle, draw] { };
\draw (1.9,4.13) node(j) [circle, draw] { };
\draw (4.05,5.68) node(f) [circle, draw] { };

\draw (4.55, 7.68) node() {$a$};
\draw (8.6, 4.75) node() {$b$};
\draw (7.05, 0) node() {$c$};
\draw (1.05, 0) node() {$d$};
\draw (-0.5, 4.75) node() {$e$};

\draw (4.55,5.68) node() {$f$};
\draw (6.7,3.9) node() {$g$};
\draw (5.85, 1.8) node() {$h$};
\draw (3.2, 1.4) node() {$i$};
\draw (1.4,3.9) node() {$j$};

\draw (a)--(b)--(c)--(d)--(e)--(a);
\draw (f)--(h)--(j)--(g)--(i)--(f);
\draw (d)--(i);
\draw (d) .. controls (1.55,0.7) and (2.25,1.4) .. (i);
\end{tikzpicture}}
\caption{Illustration of the algorithm for the Petersen graph.}
\end{figure}
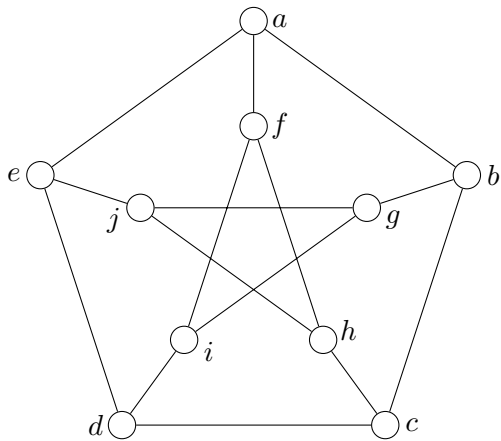
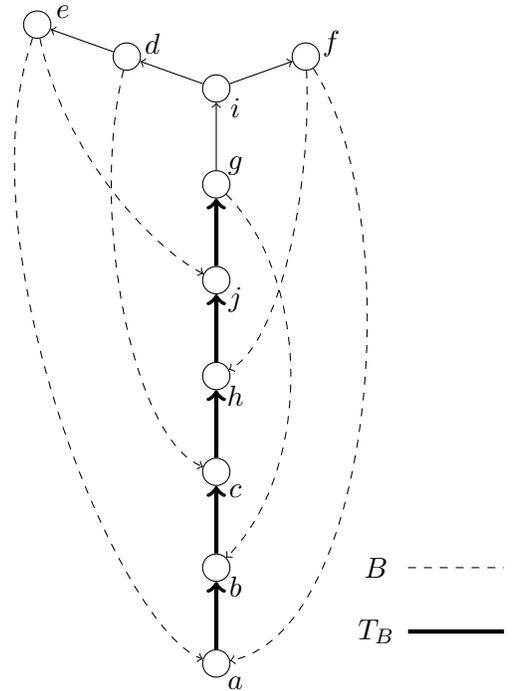
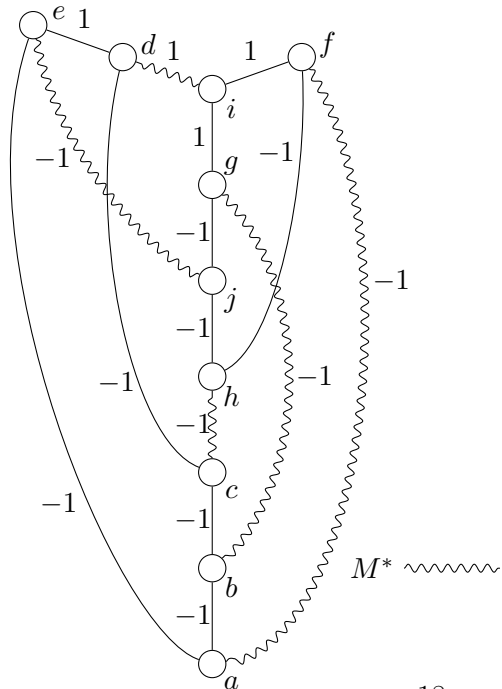
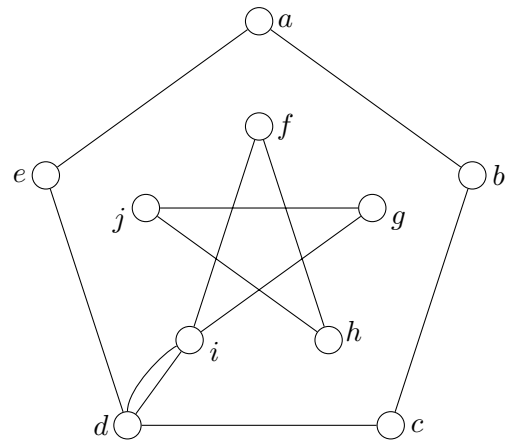

In Figure 2 we give an illustration of the algorithm when applied to the Petersen graph $G=(V,E)$ (Figure 2(a)).  Figure 2(b) shows a depth first search tree for $G$ using vertex $a$ as the root, with the back edges $B$ illustrated with dashed lines, and the corresponding edges $T_B$ illustrated with bold lines.  In Figure 2(c) the assigned edge weights $c_e$ are indicated for each edge $e\in E$, as well as the the minimum weight perfect matching $M^*$ of $G$ w.r.t. these edge weights (edges in $M^*$ are indicated with wavy lines).  Note that $c(M^*)=-3$. Also note that $R=T_B \cup B$ is the set of edges that are assigned a weight of $-1$. Finally, Figure 2(d) shows the spanning Eulerian multi-subgraph $H$ obtained by removing from $E$ all the edges of $M^*\cap R$ and adding an extra copy of the edges in $M^*\setminus R$. As can be seen, the number of edges in the final solution is $|E| + c(M^*) = 12 \leq 4n/3 - 2/3$, as required.

The running time of the algorithm described above is dominated by the time required to find a minimum cost perfect matching.  This step can be performed in $O(n(|E| + nlogn))$ time (see \cite{gabow}), which is $O(n^2logn)$ for cubic graphs.

As a result of the above analysis the next theorem follows.
\begin{theorem} \label{thm:main2}
Let $G=(V,E)$ be a bridgeless cubic multigraph with $n$ vertices. There is an $O(n^2logn)$ algorithm that finds a spanning Eulerian multi-subgraph $H$ of $G$ with at most $4n/3 - 2/3$ edges. \qed
\end{theorem}

\noindent Note that the result in Theorem \ref{thm:main2} is tight, for example, consider the graph $G$ which consists of 2 vertices joined by 3 parallel edges.  We not that for simple graphs we have not been able to find a better lower bound than $11n/9-8/9$ (see Section~\ref{sec:epilogue}).

The minimum cost perfect matching algorithm used in the above can be viewed as a more efficient derandomization of the randomized algorithm in \cite{momke} than the obvious derandomzation, which is to consider each of the perfect matchings in a convex combination of the all-$1/3$ vector in the perfect matching polytope (cf. Lemma~\ref{lem:ccofmatchings}); in \cite{momke} this convex combination is used in the interpretation as a probability distribution over the vertices of the perfect matching polytope.  In fact, given the sets $R$ and $P$ of any graph, not just cubic, these ideas are easily extended to provide a minimum cost perfect matching problem for the cubification of the graph used in \cite{momke} for general graphs, simply by using the weight function $c$ as described in
(\ref{costs}) for edges of the original graph, and setting $c_e := 0$ for all edges $e$ in the cubification which are not in the original graph. Thus, also for general graphs and the $1.461$ result in \cite{momke}, a more efficient derandomization is possible.

A final remark concerns an alternative view on how the above algorithm works. Basically, the algorithm combined a spanning tree $T^*$ with the perfect matching $M^*$. The spanning tree $T^*$ is the one obtained from the original depth first search tree $T$ by removing the edges of $M^* \cap T_B$ and replacing them by their partner back edges. The analysis shows that
\begin {eqnarray}
|M^* \cap T^*| \leq n/6 - 1/3.
\end {eqnarray}
The resulting spanning Eulerian multi-subgraph of the algorithm is then obtained by removing the edges of
$M^* \backslash T^*$ from $E$ and adding an extra copy of the edges $M^* \cap T^*$, indeed containing $n + 2(n/6 - 1/3)=4n/3 - 2/3$ edges.

In Figure 3 we show the spanning tree $T^*$ obtained in this way that corresponds to $M^*$  for the Petersen graph example of Figure 2.

\begin{figure}\label{fig:treeMatching}
\center
\begin{tikzpicture}[scale=0.8]

\draw (1.55,0) node(d) [circle, draw] { };
\draw (6.55,0) node(c) [circle, draw] { };
\draw (8.1,4.75) node(b) [circle, draw] { };
\draw (0,4.75) node(e) [circle, draw] { };
\draw (4.05,7.68) node(a) [circle, draw] { };

\draw (2.73,1.62) node(i) [circle, draw] { };
\draw (5.37,1.62) node(h) [circle, draw] { };
\draw (6.2,4.13) node(g) [circle, draw] { };
\draw (1.9,4.13) node(j) [circle, draw] { };
\draw (4.05,5.68) node(f) [circle, draw] { };

\draw (4.55, 7.68) node() {$a$};
\draw (8.6, 4.75) node() {$b$};
\draw (7.05, 0) node() {$c$};
\draw (1.05, 0) node() {$d$};
\draw (-0.5, 4.75) node() {$e$};

\draw (4.55,5.68) node() {$f$};
\draw (6.7,3.9) node() {$g$};
\draw (5.85, 1.8) node() {$h$};
\draw (3.2, 1.4) node() {$i$};
\draw (1.4,3.9) node() {$j$};

\draw[ultra thick] (a)--(b)--(c)--(d)--(e);
\draw (a)--(e);
\draw[ultra thick] (f)--(i)--(g)--(j)--(h);
\draw[ultra thick] (d) .. controls (1.55,0.7) and (2.25,1.4) .. (i);
\draw (f)--(h);

\draw decorate [decoration={name=snake, amplitude=1.5pt, segment length=5pt}] {(a)--(f)};
\draw decorate [decoration={name=snake, amplitude=1.5pt, segment length=5pt}] {(b)--(g)};
\draw decorate [decoration={name=snake, amplitude=1.5pt, segment length=5pt}] {(c)--(h)};
\draw decorate [decoration={name=snake, amplitude=1.5pt, segment length=5pt}] {(d)--(i)};
\draw decorate [decoration={name=snake, amplitude=1.5pt, segment length=5pt}] {(e)--(j)};

\draw (8.5,2.5) node() {$T^*$};
\draw[ultra thick] (9,2.5)--(10.5,2.5);

\draw (8.5,1.5) node() {$M^*$};
\draw decorate [decoration={name=snake, amplitude=1.5pt, segment length=5pt}] {(9,1.5)--(10.5,1.5)};

\end{tikzpicture}
\caption{Final correct matching $M^*$ and spanning tree $T^*$ pairing for the Petersen graph.}
\end{figure}
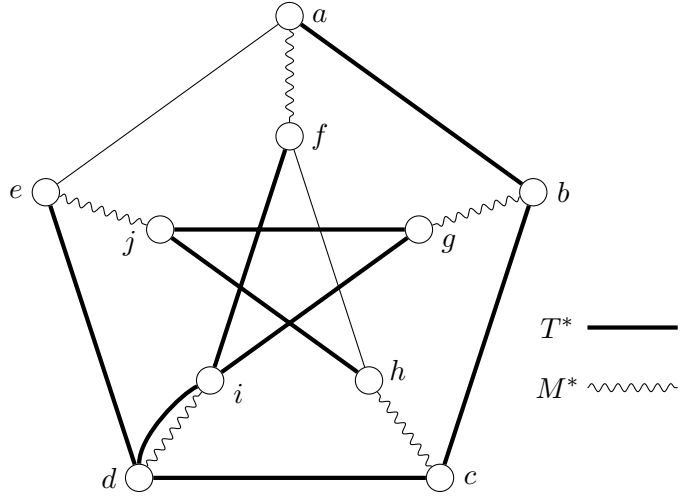

Seen in this way it becomes directly clear that the solution that the algorithm produces is in fact a cycle cover together with double edges to connect them.
This also immediately implies a bound on the number of cycles in the cycle cover:
\begin{corollary}
  Given any bridgeless cubic multigraph, there exists a cycle cover with at most $\lfloor n/6+2/3 \rfloor$ cycles.
\end{corollary}

\subsection{Subcubic bridgeless multigraphs.}\label{sec:subcubicwts}

Let $G=(V,E)$ be a bridgeless subcubic multigraph (i.e. all vertices in $G$ have degree 2 or 3) with $n$ vertices. Let $V_3$ and $V_2$ denote the sets of vertices of degree 3 and 2, respectively, and let $n_3 = \left| {V_3} \right|$ and $n_2 = \left| {V_2} \right|$.  We will prove that the equivalent form of Theorem \ref{thm:main2} for this type of graph also holds. Note that also for subcubic graphs the bound in the theorem is tight (consider a graph that consists of three paths of the same length joining 2 vertices).
\begin{theorem} \label{thm:main2subcubic}
 Let $G=(V,E)$ be a bridgeless subcubic multigraph with $n$ vertices. There is a $O(n^2logn)$ algorithm that finds a spanning Eulerian multi-subgraph $H$ of $G$ with at most $4n/3 - 2/3$ edges.
\end{theorem}
\begin{proof}
To begin, we replace every path $Q$ consisting of degree 2 vertices in $G$ by a single edge $e_Q$ to obtain a cubic bridgeless multigraph $G^{\prime} = (V_3,E^{\prime})$.  For every such path $Q$, let $Q_2$ be the set of degree 2 vertices that lie on it.  Note that $V_2$ is the union of all sets $Q_2$ over all the paths $Q$.  We now proceed as in the proof for cubic graphs above: We find the tree $T$ and set of edges $R$ for $G^{\prime}$, as well as a minimum weight perfect matching $M^*$, however this time we will use a different edge weight function:  For every edge $e_Q \in E^{\prime}$, let
\begin{displaymath}
c_{e_Q} = \left \{\begin{array}{rl}
                   \left|{Q_2}\right|-1 & \mbox{ if }e \in R, \\
                   \left|{Q_2}\right|+1 & \mbox{ if }e \in E^{\prime}\setminus R.
                   \end{array}
           \right.
\end{displaymath}

\noindent Using this weight function, by (\ref{costM}) and Theorem \ref{thm:NadPul} we have that
\begin{eqnarray} \label{costMsubcubic}
c(M^*) \leq  c(E^{\prime})/3 = (\sum_{e_Q \in E^{\prime}}{\left|{Q_2}\right|/3}) -(n_3/6 + 2/3) = n_2/3 -n_3/6 -2/3.
\end{eqnarray}

Now go back to the original graph $G$.  For every edge $e_Q$ (in $G^{\prime}$) that is in $M^*\setminus R$, add an extra copy of every edge in the corresponding path $Q$ to $E$, and for every edge $e_Q$ that is in $M^* \cap R$, take one edge away from path $Q$, and add an extra copy of every other edge in $Q$.  In this way, we are adding exactly $c(M^*)$ edges to $E$. As in Section \ref{sec:3.1}, this new graph will be a spanning Eulerian multi-subgraph $H$.  The number of edges in $H$ is
\begin{eqnarray*}
|E|  + c(M^*) \leq 1/2(3n_3 + 2n_2) + (n_2/3 -n_3/6 -2/3) = 4n/3 - 2/3,
\end{eqnarray*}
as required, and the theorem follows.
\qed
\end{proof}\smallskip

\section{Graphs with bridges}\label{sec:bridges}
We extend the analysis to any subcubic graph by studying bridges. Deleting the bridges of a graph splits it into separate components each of which is either a single vertex or a subcubic bridgeless graph. Let $h$ be the number of bridges in a graph and $s$ the number of vertices incident to more than one bridge.
\begin{theorem} \label{thm:subcubic_bridges}
For a subcubic graph with $h$ bridges, a TSP tour of length at most $(4/3)(n+h) - (2/3)(s+1)$ can be constructed.
\end{theorem}
\begin{proof} Removing the bridges yields $h+1$ bridgeless components, $s$ of them being single vertex components. Thus, there are $h+1-s$ subcubic components, for each of which we can find a TSP tour of length at most $(4/3)n'-2/3$, where $n'$ is the number of vertices in the component. Adding two copies of each bridge yields a TSP tour of length at most $(4/3)(n-s) - (2/3)(h+1-s) +2h = (4/3)(n+h) - (2/3)(s+1)$.
\qed
\end{proof}\smallskip

Since for a graph with $h$ bridges $n+2h-s$ is a lower bound both for the number of edges on an optimal tour and for the optimal solution of the SER  we obtain the following corollary.
\begin{corollary}\label{4over3subcubic}
For graph-TSP on subcubic graphs, there exists a polynomial-time $4/3$-approximation algorithm, and the integrality gap for SER is at most $4/3$.
\end{corollary}
\begin{proof}
Using $n+2h-s$ as lower bound we have
\begin{eqnarray*}
\frac{(4/3) (n+h) - (2/3)(s+1)}{(n+2h-s)} = 4/3 - \frac{(4/3)h - (2/3)s - 2/3}{n+2h-s},
\end{eqnarray*}
which is at most 4/3, since $s\leq h$ (except for the case where $s=n$, but then the graph is a tree and the corollary is trivially true).
\qed
\end{proof}\smallskip

\section{Epilogue}
\label{sec:epilogue} Very recently, remarkable progress has been made on the approximability of graph-TSP. In the table below we show the present  state of knowledge. It contains: (1st column) lower
bounds on the length of graph-TSP tours on $n$ vertices, for $n$ large enough,  (2nd column) upper bounds on them that we know how to construct,
(3rd column) lower bounds on the integrality gap of SER, (4th column) upper bounds on the integrality gap of SER, and (last column) upper bounds on the best possible approximation ratio. The bounds apply to bridgeless graphs, because they are the crucial ones within the classes. All lower bounds hold for simple graphs.

\begin{eqnarray*}
\begin{array}{l|lllll}
 &\text{TSP lb} & \text{TSP ub}& \text{SER lb}& \text{SER ub} &  \text{Approx.} \\ \hline
\text{General graphs}  &2n-4&2n-2&4/3&1.461&1.461 \\
\text{Subcubic graphs} &4n/3-2/3& 4n/3-2/3 & 4/3 & 4/3 & 4/3 \\
\text{Cubic graphs} &11n/9 -8/9&4n/3-2 & 7/6 & 4/3 & 4/3 \\
\end{array}
\end{eqnarray*}

The graph-TSP  lower bound for general graphs is given by the complete bipartite graph $K_{2,n-2}$ (on 2 and $n-2$ vertices). The  graph-TSP lower bound for cubic graphs we prove in a lemma. Notice that if we do not restrict to simple graphs then the graph with two vertices and three edges yields a lower bound of $2=(4/3)2-2/3$ for the cubic case.

\begin{figure}\label{fig:LBCubic}
\center
\begin{tikzpicture}[scale=0.65]

\draw (10,0) node(a1) [circle, draw] { };
\draw (10,1) node(a2) [circle, draw] { };
\draw (10,1.75) node(a3) [circle, draw] { };
\draw (10,2.75) node(a4) [circle, draw] { };
\draw (10,5) node(a9) [circle, draw ] { };
\draw (10,6) node(a10) [circle, draw] { };
\draw (10,6.75) node(a11) [circle, draw] { };
\draw (10,7.75) node(a12) [circle, draw] { };

\draw (8,0.5) node(b1) [circle, draw] { };
\draw (8,2.25) node(b2) [circle, draw] { };
\draw (8,5.5) node(b5) [circle, draw] { };
\draw (8,7.25) node(b6) [circle, draw] { };

\draw (6,1.375) node(c1) [circle, draw] { };
\draw (6,6.375) node(c3) [circle, draw] { };

\draw (3,3.875) node(d) [circle, draw] { $s$};

\draw (12,0.5) node(e1) [circle, draw] { };
\draw (12,2.25) node(e2) [circle, draw] { };
\draw (12,5.5) node(e5) [circle, draw] { };
\draw (12,7.25) node(e6) [circle, draw] { };

\draw (14,1.375) node(f1) [circle, draw] { };
\draw (14,6.375) node(f3) [circle, draw] { };

\draw (17,3.875) node(g) [circle, draw] { $t$};

\draw (a1)--(a2);
\draw (a3)--(a4);
\draw (a9)--(a10);
\draw (a11)--(a12);

\draw (a1)--(b1)--(a2)--(e1)--(a1);
\draw (a3)--(b2)--(a4)--(e2)--(a3);
\draw (a9)--(b5)--(a10)--(e5)--(a9);
\draw (a11)--(b6)--(a12)--(e6)--(a11);

\draw (b1)--(c1)--(b2);
\draw (b5)--(c3)--(b6);

\draw (e1)--(f1)--(e2);
\draw (e5)--(f3)--(e6);

\draw (d)--(c1);
\draw (d)--(c3);

\draw (g)--(f1);
\draw (g)--(f3);

\draw (d)--(g);

\end{tikzpicture}
\caption{Family of cubic graphs for which the optimal graph-TSP tour has length  $11n/9 - 8/9$.}
\end{figure}
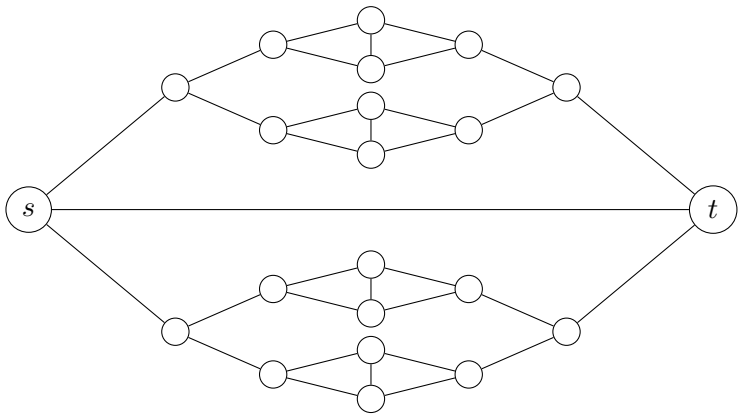

\begin{lemma}
For any $n_1$ there is cubic bridgeless graph on $n>n_1$ vertices such that the optimal tour has length at least $11n/9-8/9$.
\end{lemma}
\begin{proof}
Take two complete binary trees and connect their leaves as in Figure 4 and add an edge between the two roots $s$ and $t$. Let $2k+2$ be the distance from $s$ to $t$ not using edge  $st$.  Denote the corresponding graph by $F_k$. The example shows $F_2$. In general $k\ge 1$ and $n=6\cdot2^k-2$.
 Now let us compute an optimal TSP tour. Let $T(k),P(k)$ be the length of the shortest connected Eulerian subgraph in $F_k$ using edge $st$  respectively $0$ and $1$ times.
 Then, $T(1)=10$ and $P(1)=12$. Consider a minimum spanning connected Eulerian subgraph in $F_k$.  If it does not contain edge $st$, then the Eulerian subgraph either contains exactly one copy of each of the four edges incident to $st$ or three of these four edges doubled. In the first case
$T(k)=4+2(P(k-1)-1)$ and in the latter we have $T(k)=6+2T(k-1)$. Hence, $T(k)=\min\{6+2T(k-1),2+2P(k-1)\}$.

If the Eulerian subgraph does contain edge $st$ then is easy to see that $P(k)=5+T(k-1)+(P(k-1)-1)=4+T(k-1)+P(k-1)$.
Given the initial values $T(1)=10$ and $P(1)=12$ the values that follow from these equations are uniquely defined. One may verify that the following functions satisfy the equations.
\[
\begin{array}{lll}
T(k)=22/3\cdot 2^k-14/3,& P(k)=22/3\cdot 2^k-8/3&\text{for odd }k,\\
T(k)=22/3\cdot 2^k-10/3,& P(k)=22/3\cdot 2^k-10/3&\text{for even }k.
\end{array}
\]
For even $k$ the length of the optimal tour is $22/3\cdot 2^k-10/3=11n/9-8/9$.\qed
\end{proof}\smallskip

We believe that for simple cubic graphs there exists a polynomial-time algorithm with approximation ratio strictly less than 4/3. In fact, the problem is not known to be APX-hard.

The lower bound of $7/6$ on the integrality gap for cubic graphs is attained by the following graph. Connect two points by three equally long paths. Then replace every vertex of degree 2 by a 4-cycle with a chord so as to make the graph cubic.

Of course, the main research challenges remain to prove Conjecture~\ref{mainconj} or to show a $4/3$-approximation algorithm.
For general metric TSP even an approximation ratio strictly less than 3/2 is still wide open.  For graph-TSP M\"omke and Svensson \cite{momke} made a promising and important step. It seems that especially good lower bounds on optimal solutions are still lacking.

\begin{small}

\end{small}

\end{document}